\documentclass[11pt,a4paper]{article}
\usepackage{fourier} 
\usepackage{amsmath,amsbsy,bm,amssymb,enumerate,mathrsfs}
 
\usepackage{graphicx}

\usepackage[utf8]{inputenc} 
\usepackage{utf8math}

\usepackage{anysize} 

\usepackage{subfigure}

\providecommand{\setN}{\mathbb{N}}

\newcommand {\R}	  {\mathbb{R}}

\DeclareMathOperator{\sign}{sign}
\renewcommand{\epsilon}{\varepsilon}

\renewcommand{\leq}{\leqslant}
\renewcommand{\geq}{\geqslant}

\usepackage{amsthm}
\theoremstyle{plain}
\newtheorem{theorem}{Theorem}

\newtheorem{lemma}[theorem]{Lemma}

\theoremstyle{definition}

\newtheorem{remark}{Remark}

\title{Proximity results and faster  algorithms for Integer Programming  using the Steinitz Lemma} 

\author{Friedrich Eisenbrand \\
 EPFL\\
 Switzerland\\
 {\small \texttt{friedrich.eisenbrand@epfl.ch}}
\and
 Robert Weismantel \\
 ETH Zürich\\
 Switzerland\\
 {\small \texttt{robert.weismantel@ifor.math.ethz.ch}}
}
\date{\today}

\begin{document}

\title{\Large Proximity results and faster  algorithms for Integer Programming  using the Steinitz Lemma} 
\author{Friedrich Eisenbrand\thanks{EPFL, Lausanne,  Switzerland} 
\and
Robert Weismantel\thanks{
 ETH Zürich,  Switzerland}}
\date{\today}

\maketitle

\begin{abstract}
  \noindent 
  We consider integer programming problems in standard form 
  $\max \{c^Tx : Ax = b, \, x≥0, \, x ∈ ℤ^n\}$ 
  where $A ∈ ℤ^{m ×n}$,
  $b ∈ ℤ^m$
  and $c ∈ ℤ^n$.
  We show that such an integer program can be solved in time
  $(m ⋅ Δ)^{O(m)} ⋅ \|b\|_∞^2$,
  where $Δ$
  is an upper bound on each absolute value of an entry in $A$. 
  This improves upon the longstanding best bound  of
  Papadimitriou (1981) of $(m⋅Δ)^{O(m^2)}$,  where in addition,  the absolute values of the entries of $b$ also need to be bounded by $Δ$. 
%  and addresses an open problem raised  by Fomin. 
  Our result relies on a lemma of Steinitz that states that a set of
  vectors in $ℝ^m$ that is 
  contained in the unit ball of a norm and that sum up to zero can be ordered 
  such that all partial sums are of norm bounded by $m$.

  We also use the Steinitz lemma to show that the $\ell_1$-distance
  of an optimal integer and fractional solution,  
  also under the presence of upper bounds on the variables, is bounded
  by $m ⋅ (2\,m⋅Δ+1)^m$.
  Here $Δ$
  is again an upper bound on the absolute values of the entries of
  $A$.
  The novel strength of our bound is that it is independent of $n$.

  We provide evidence for the significance of our bound by applying it to general knapsack problems where we obtain
   structural and algorithmic results that improve upon  the recent literature. 
\end{abstract}

\section{Introduction}
\label{sec:introduction}

Many algorithmic problems,  most notably problems from \emph{combinatorial optimization} and the \emph{geometry of numbers} can be formulated as an \emph{integer linear  program}. This is an optimization problem of the form  
\begin{equation}
  \label{eq:2}
  \max \{ c^T x : Ax = b, \, x ≥0, \, x ∈ ℤ^n \}  
\end{equation} 
where $A ∈ ℤ^{m×n}$, $b ∈ℤ^m$ and $c ∈ ℤ^n$.
An integer program as we describe it above is in \emph{(equation) standard form}. Any integer program in \emph{inequality form}, i.e.,  $\max\{ c^Tx : Ax ≤ b, \, x ∈ℤ^n\}$ can be transformed into an integer program in standard form by duplicating variables and introducing \emph{slack variables}. 
 Unlike \emph{linear programming}, integer programming is NP-complete~\cite{BoTr76}. %Nevertheless integer programming solvers are nowadays capable of solving large instances efficiently in practice. 

Lenstra~\cite{Lenstra83} has shown that an integer program in inequality form, with a fixed number of variables can be solved in polynomial time. A careful analysis of his algorithm shows a time bound of $2^{O(n^2)}$ times a polynomial in the length of the input that contains binary encodings of numbers. This has been improved by Kannan~\cite{Kannan87} to $2^{O(n \log n)}$ which is the best asymptotic upper bound  on the exponent of $2$ in 30 years. The question whether this can be improved to $2^{O(n)}$ belongs to one of the most prominent mysteries in the theory of algorithms. The current record on the constant hidden in the $O$-notation in the exponent  is held by Dadush~\cite{dadush2012integer}.

%\subsection*{Dynamic programming algorithms} 
\medskip
\noindent 
Papadimitriou~\cite{MR677087} has provided  an algorithm for integer programs in standard form that is, in some sense, complementary to the result of Lenstra and its improvement of Kannan. He considered the case of an integer program~\eqref{eq:2} in which the entries of $A$ and $b$ are bounded by $Δ$ in absolute value. His algorithm is \emph{pseudopolynomial} if $m$ is fixed and is thus a natural generalization of pseudopolynomial time algorithms to solve unbounded knapsack problems~\cite{MR0378463}.

The algorithm is based on  dynamic programming and can be briefly described as follows. 
First, one shows that,  if~\eqref{eq:2} is feasible and bounded, then \eqref{eq:2} has
an optimal solution with components bounded by $U = (n+1) (m \cdot
Δ)^m$.
The dynamic program is a maximum weight path problem on the (acyclic)
graph with nodes
\begin{displaymath}
V = \{0,\ldots,n\} × \{ -n⋅Δ⋅U, \ldots, n⋅Δ⋅U\}^m 
\end{displaymath}
where one has an arc from $(j,b')$ to $(j+1,b'')$ if $b'' - b' = k
\cdot a^{(j+1)}$ for some $k \in \setN_0$ and where $a^{(j+1)}$ is the
$j+1$-st column of $A$. The weight of this arc is $k \cdot
c_{j+1}$. The optimum solution corresponds to a longest path to the
vertex $(n,b)$. The running time of this algorithm is linear in the size of the graph. 
The number of nodes of this graph  is equal to 
$(n+1) ⋅ ( 2n Δ U +1)^m$ and this is at least 
$(m Δ)^{m^2}$. 
The upper bound on the running time, as stated  in~\cite{MR677087} is  
\begin{equation}
  \label{eq:8}
  O(n^{2m+2} ⋅ (mΔ)^{(m+1)(2m+1)}). 
\end{equation}

%If all entries of $A$ are bounded by $Δ$ in absolute value, then $A$ can have at most $(2Δ+1)^m$ different columns. Without loss of generality, we can assume that $A$ does not have repeated columns which implies that $n = Δ^{O(m)}$. The running time of Papadimitriou's algorithm is thus bounded by 
%\begin{displaymath}
%  O(Δ^{O(m^2)}).   
%\end{displaymath}

%\todo{various configurations have various weights}

\subsection{Contributions of this paper} 
\label{sec:contr-this-paper}

We present new structural and algorithmic results concerning integer programs in standard form~\eqref{eq:2} using the Steinitz lemma, see Section~\ref{sec:contr-this-paper} below.

\begin{enumerate}[a)]
\item \label{item:4}
We show that the integer program~\eqref{eq:2} can be solved in time
\begin{displaymath}
  (m⋅Δ)^{O(m)} ⋅ \|b\|_∞^2 
\end{displaymath}
where $Δ$ is an upper bound on the entries of $A$ only. This improves upon the $(m⋅Δ)^{Ω(m^2)}$ running time of the algorithm of Papadimitriou. Recall that in the setting of Papadimitriou the entries of  $b$  are bounded by $Δ$ as well.  This improvement  addresses an open problem raised by Fomin et al.~\cite{fomin2016fine,lewenstein_et_al:DR:2017:7037}. 
\end{enumerate}
% Since the variables are not bounded from above, we can assume that the matrix $A$ does not have repeated columns and thus $n ≤ (2⋅Δ+1)^m$. This shows that a running time of $Δ^{O(m)}$ is optimal, since, in the worst case, $A$ has that many columns and the complete input must be read. 

\medskip 
\noindent 
We then consider integer programs of the form
\begin{displaymath}
  \max \{c^Tx : Ax = b, \, 0 ≤x≤u, \, x ∈ ℤ^n\} 
\end{displaymath}
where $A ∈ℤ^{m ×n}$,
$b ∈ ℤ^m$,
$u ∈ℕ^n$,
and $c ∈ℤ^n$
and $|a_{ij}| ≤ Δ$
for each $i,j$.
Thus we allow the variables of integer program~\eqref{eq:2} to be
bounded from above by $0≤x≤u$
for some $u ∈ ℕ^n$.
In this setting, we show the following.

\begin{enumerate}[a)]
 \setcounter{enumi}{1}
\item We provide new  bounds on  the distance of an optimal vertex $x^*$ of the LP-relaxation and an optimal solution  of the integer program itself. More precisely, we show that there exists an optimal solution $z^*$ of the integer program such that 
  \begin{displaymath}
    \|z^* - x^* \|_1 ≤ m ⋅ (2⋅m⋅Δ+1)^m 
  \end{displaymath}
holds. 
% This is an improvement by a factor of $n^2$ over the bound on this distance that is derived, using a bound of 
A classical bound of 
Cook et al.~\cite{CookGerardsSchrijverTardos86} implies, in the standard-form setting, $\|z^* - x^*\|_∞ ≤ n ⋅ (\sqrt{m} ⋅Δ)^m$ and thus $\|z^* - x^*\|_1 ≤ n^2 ⋅ (\sqrt{m} ⋅Δ)^m$. 
Thus our bound, which is independent of $n$, is an improvement by a factor of $n^2$ for integer programs in standard form and fixed $m$. \label{item:5}
\item We use this to  generalize  a recent bound on the absolute integrality gap for the case $m=1$ by Aliev et al.~\cite{aliev2016integrality} that states that $c^T(x^* - z^*) ≤ \|c\|_∞ ⋅ 2 ⋅Δ $. Our distance bound shows that the absolute integrality gap is bounded by $\|c\|_∞ ⋅ O(m)^{m+1} ⋅  O(Δ)^m$. \label{item:7}
\item Our new distance bound yields an algorithm for integer programs in standard form that runs in time 
  \begin{displaymath}
    n ⋅  O(m)^{(m+1)^2} ⋅ O(Δ)^{m⋅(m+1)} ⋅ \log^2 (m⋅Δ)
  \end{displaymath}
For the unbounded and bounded knapsack problems where all items are of weight $Δ_a$ at most, we obtain algorithms that run in time $O(n⋅Δ_a^2)$ and $O(n^2⋅Δ_a^2)$ respectively. This is an improvement by a factor of $n$ to the so far best bounds for this problem by  Tamir~\cite{tamir2009new}. 
\label{item:6}
% \item Finally we show that there exists an optimal solution $z^*$ of an integer program of the form~\eqref{eq:2} that has support bounded by $2m ⋅ \log(4mΔ)$. A previous bound of Eisenbrand and Shmonin~\cite{ES05}, see also \cite{aliev2017sparse} for a recent improvement thereof,  required the infinity norm of $c$ to be bounded by $Δ$ as well. This result on the support is similar in spirit to its linear programming counterpart, where the objective function does not affect the bound as well. 
\end{enumerate}

\noindent
Our techniques have been recently refined  by Jansen and Rohwedder~\cite{jansen2018integer} who obtained better constants in the exponent of the running time of integer programs without upper bounds. We also want to mention a recent tight lower bound for integer programming. Knop et al.~\cite{knop2018tight}  prove that even for $\{0,1\}$-matrices, the running time of our algorithm is probably optimal. In a nutshell, an algorithm with better asymptotic running time in the exponent for unbounded integer programs would contradict the exponential time hypothesis. 
This improves the lower bounds of Fomin et al.~\cite{fomin2016fine}.

\subsection{The Steinitz lemma} 
\label{sec:lemma-steinitz}

Our algorithms and structural results rely on a Lemma of
Steinitz~\cite{steinitz1913bedingt} that we now describe. Here $\|\cdot\|$ denotes  an arbitrary norm of $ℝ^m$. 
 
\begin{theorem}[Steinitz (1913)]
	\label{thr:3}
        Let
        $x_1,\ldots,x_n \in \R^m$
        such that 
        \begin{displaymath}
          \sum_{i=1}^n x_i = 0    \quad \text{ and } \quad  \|x_i\| \leq 1 \, \text{ for each }i. 
        \end{displaymath}
        There exists a permutation $π ∈ S_n$
        such that all partial sums satisfy 
	$$\|\sum_{j=1}^k x_{\pi(j)}\| \leq c(m) \; \text{ for all } k=1,\ldots,n.$$ Here $c(m)$ is a constant depending on $m$ only. 
\end{theorem}
Steinitz showed $c(m)≤ 2m$, see also~\cite{sevast1997steinitz,barany2008power}. 
It was later shown by Sevast'anov~\cite{sevastianov1978approximate,sevast1997steinitz} that the constant   $c(m)≤m$. This is tight for asymmetric norms i.e., general gauge functions~\cite{grinberg1980value}. However, this bound is not optimal for symmetric norms. In particular, Banaszczyk proved $c(m) ≤ m - 1 + 1/m$, see~\cite{sevast1997steinitz}. It is a wide open question to understand the Steinitz constant for $\ell_p$-norms for $p≥2$. It is conjectured that the Steinitz constant should be $O(\sqrt{m})$ for the $\ell_\infty$-norm~\cite{barany2008power}. A proof of this conjecture or any asymptotic improvement would directly improve the bounds provided in this manuscript and would provide tightness results in a variety of settings. 
\begin{figure*}
\label{fig:1}
\begin{center}
  \includegraphics{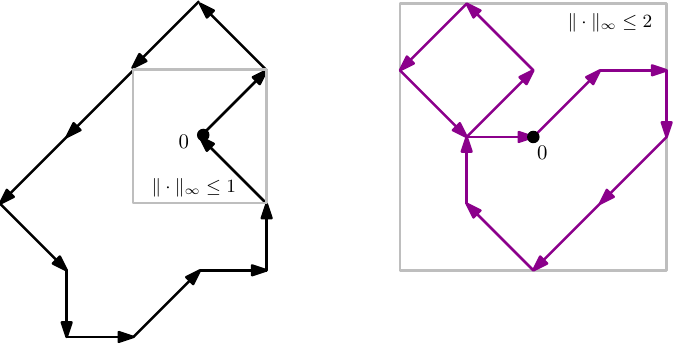}
\end{center}
\caption{An example of a re-ordering satisfying the Steinitz bound for the $\ell_∞$-norm. The vectors on the left have $\ell_∞$-norm at most one and summ up to zero. These vectors are rearranged on the right such that the partial sums have $\ell_∞$-norm bounded by $2$.}
\end{figure*}
The proof of the Steinitz lemma with constant $c(m) = m$  is based on LP-techniques~\cite{grinberg1980value} and can be quickly summarized as follows. One constructs sets $A_n ⊃A_{n-1}⊃ \cdots ⊃A_m$ where $A_n = \{1,\dots,n\}$ and $|A_k|=k$  for each $k$ such that the following linear system  which is described by $A_k$ with variables $λ_i, \, i ∈A_k$  is feasible for each $k$: 
\begin{equation}
  \label{eq:10} 
  \begin{array}{lcl}
    ∑_{i ∈A_k} λ_i x_i     &=& 0 \\
    ∑_{i ∈A_k} λ_i & = & k-m \\
    0 ≤ λ_i ≤1,  \,&& i ∈A_k.    
  \end{array}  
\end{equation}
For any permutation $π$ with $\{π(i)\} = A_i \setminus A_{i-1}$ for $i=n,\dots,m+1$ one has then for any $k≥m$
\begin{eqnarray*}  
\|  ∑_{i=1}^k x_{π(i)} \| & = & \|  ∑_{i∈ A_k} x_{i} \| \,
                        =   \|  ∑_{i∈ A_k} (1 - λ_i) \, x_{i} \|\\
                       & ≤ & ∑_{i∈ A_k} (1 - λ_i) 
                       \, =  m.
\end{eqnarray*}
In the  inequality, we used $\|x_i\| ≤1$ for each $i$ and in the first and second  equation  we used~\eqref{eq:10}. 
The sets $A_k$ are constructed inductively as follows. $A_n = \{1,\dots,n\}$. If $A_k$ has been constructed, where $k>m$, one first notes that the system~\eqref{eq:10} is of course also solvable if the right-hand-side  $k-m$ of the second constraint is replaced by $k-1-m$. Once this replacement has been done, one observes that~\eqref{eq:10} consists of  $m+1$ equations and the inequalities $0≤λ≤1$. A vertex solution of~\eqref{eq:10} has thus at most $m+1$ fractional entries that sum up to a value less than $m+1$. A vertex solution of~\eqref{eq:10} must therefore have one entry equal to zero. Otherwise the components of the vertex sum up to a value larger than $k-1-m$. The set $A_{k-1}$ is now the set $A_k$ from which the index corresponding to the zero in the vertex solution has been removed.

The reader will notice some resemblance in spirit to the proof of the Beck-Fiala theorem in Discrepancy Theory~\cite{beck1981integer,matousek2009geometric}.
Discrepancy techniques have given improvements to the Steinitz and Beck-Fiala problem, when one allows a weak dependence on the number of vectors. In particular, Banaszczyk~\cite{banaszczyk1998balancing} proved a Beck-Fiala bound of $O(\sqrt{t \log n})$ for set systems on $n$ elements with sparsity $t$. He also derived an $\ell_\infty$ Steinitz bound of $O(\sqrt{m \log n})$ for $n$ vectors in dimension $m$~\cite{banaszczyk2012series}. For constructive versions of these bounds we refer the interested reader to~\cite{bansal2016algorithm,bansal2017algorithmic}.

We are not the first to apply the Steinitz lemma in the context of
integer programming. Dash et al.~\cite{dash2012master} have shown that
an integer program~\eqref{eq:2} can be solved in pseudopolynomial time
if a certain parameter of the number of rows $τ$ 
is a function of $m$, i.e., $τ=τ(m)$.
The interesting aspect of their algorithm is that it relies on linear
programming techniques only. The number of inequalities in their
linear program is bounded by  an exponential in $τ(m)$.
Buchin et al.~\cite{buchin2012vectors} have shown that
$m^{m/2-o(m)}≤τ(m) ≤ m^{m+ o(m)}$
which then yields  an algorithm for integer programming that is
pseudopolynomial for fixed $m$
but doubly exponential in $m$.
Their upper bound on $τ(m)$
is proved via the Steinitz lemma. We take a different path in applying
the Steinitz lemma. We use it to derive more efficient dynamic programming formulations directly and indirectly via new proximity results between integer and linear programming optimal solutions.

% Let us comment on the special case where $m=1$. In this special case one can easily verify that $c(1)=1$. To see this, consider real numbers $x^1,\ldots,x^t \in [-1,1]$ such that 
% $\sum_{i=1}^t x^i=0$. Start with $x^1$ and initialize an auxilary variable $\alpha = x^1$. Let $X$ be the multiset $X = \{x^2,\ldots,x^t\}$.  

% While  $X \neq \emptyset$, perform the following steps:
% \begin{itemize}
% 	\item If $\alpha > 0$, pick an element $x^i \in X$ such that $x^i \leq 0$. 
% 	 \item If $\alpha \leq 0$, pick an element $x^i \in X$ such that $x^i > 0$. 
% \end{itemize}
% Update $\alpha:= \alpha + x^i$ and $X = X \setminus \{x^i\}$.

% This procedure generates a permutation with the desired property.
% Indeed, all values $\alpha$ that are generated in the course of this algorithm  lie in the interval $[-1,1]$. 
% This explains Theorem \ref{thr:3} for $m=1$. 

\section{A faster dynamic program}
\label{sec:fast-dynam-progr}

We now describe a dynamic programming approach to solve~\eqref{eq:2} that is based on the Steinitz-type-lemma (Theorem~\ref{thr:3}) and which is more efficient than the original algorithm of Papadimitriou~\cite{MR677087}.    Let us first consider the feasibility problem, i.e., we have to decide whether there exists a non-negative integer vector $z^* ∈ ℤ^n_{≥0}$ such that $A z^* = b$ holds. The solution $z^*$ gives rise to a sequence of vectors $v_1,\dots,v_t$ such that each $v_i$ is a column of $A$ and 
\begin{equation}
  \label{eq:st1}
  v_1+ \cdots +v_t = b. 
\end{equation}
The $i$-th column of $A$ appears $z^*_i$ times  on the left of equation~\eqref{eq:st1}  and $t = \|z^*\|_1$. This equation can be re-written as 
\begin{equation}
  \label{eq:29}
  (v_1 - b/t)+ \cdots +(v_t - b/t) = 0.
\end{equation}
Observe that the infinity norm of  each $v_i - b/t$ is at most $2\, Δ$. 
The Steinitz-type-lemma implies that there exists a permutation $π$ of the numbers $1,\dots,t$  such that all partial sums of the sequence 
\begin{equation}
  \label{eq:3}
   v_{π(1)} - b/t,\dots, v_{π(t)} -b/t
\end{equation}
have infinity norm at most $2\, m ⋅ Δ$.  
In other words,  for each $j ∈ \{1,\dots,t\}$ one has 
\begin{equation}
  \label{eq:st2}
  \|v_{π(1)} + \cdots +  v_{π(j)} - ({j}/{t}) ⋅ b \|_∞ ≤ 2\, m ⋅Δ. 
\end{equation}
This implies that each partial sum of the sequence 
\begin{displaymath}
  v_{π(1)}, \dots, v_{π(t)} 
\end{displaymath}
is contained in the set  $\mathscr{S} ⊆ ℤ^m$  
that consists of all points $x ∈ ℤ^m$ for which there exists a $j ∈ \{1,\dots,t\}$ with 
\begin{equation}
  \label{eq:30}
  \|x - (j/t) ⋅b\|_∞ ≤ 2\, m ⋅Δ.
\end{equation}
This set $\mathscr{S}$
is as large as the number of integer points at distance at most $2 ⋅ m \,Δ$ from the line segment connecting $0$ and $b$. We now argue that this number is bounded from above by 
$
|\mathscr{S}| ≤  (4\, m ⋅Δ +1)^m ⋅ \|b\|_1. 
$ 
Let $f ∈ ℝ^m$ be any point. Since $2mΔ$ is an integer, the integer points at distance at most $ 2m Δ$ from $f$ are 
contained in the set of integer points at distance at most $ 2m Δ$ from $⌊f ⌋$. 
Therefore, an upper bound on $|\mathscr{S}|$ is the number of different integer vectors that can be obtained by rounding a point on the line-segment $(0,b)$ times $(4mΔ+1)^m$. The number of rounded integer points is at most $\|b\|_1$.  

The partial sums of 
\begin{equation}
  \label{eq:31}
  v_{π(1)},\dots,v_{π(t)} 
\end{equation}
 correspond to the nodes of a directed walk from $0$ to $b$ in the   digraph $D = (\mathscr{S},A)$ where one has a directed arc $xy ∈A$  from $x ∈ \mathscr{S}$ to $y∈ \mathscr{S}$  if $y-x$ is a column of $A$. 
If there exists a path from $0$ to $b$  in this digraph $D$ on the other hand, then the arcs of the path define a multiset of columns of $A$ summing up to $b$. 

How fast is this approach to solve the integer feasibility problem?  
%In our analysis, we will not use the crude bound $n = O(Δ^m)$ but let $n$ enter the running time.  % in this way, we can benchmark our running time with various results from the literature apart form our primary benchmark, which is the first dynamic program for integer programming of Papadimitriou~\cite{MR677087}. 
The number of vertices $|\mathscr{S}|$
of the digraph is equal to $(4\,m⋅Δ+1)^m ⋅ \|b\|_1$.
The number of arcs $|A|$
is bounded by $|\mathscr{S}|⋅n$.  % = O(m⋅Δ)^m ⋅ n ⋅ \|b\|_1 = O(m⋅Δ)^{m+1} \|b\|_1$, where we assume that $A$ has no repeated columns and thus $n ≤ (2\, Δ+1)^m$. 
The integer feasibility problem is an unweighted single-source shortest path problem that can be solved with breadth-first-search in linear time~\cite{AHU74,CLRS2001}.  Consequently,  the integer feasibility problem in
standard form~\eqref{eq:2} can be solved in time 
\begin{displaymath}
 |\mathscr{S}|⋅n =  O(m⋅Δ)^m ⋅ \|b\|_1 ⋅ n.
\end{displaymath}
%where we assume that $A$ has no repeated columns and thus $n ≤ (2\, Δ+1)^m$.  

\begin{theorem}
  \label{thr:1}
  Let  $A ∈ ℤ^{m ×n}$
  and $b ∈ ℤ^m$ be given and suppose that each absolute value of an entry of
  $A$ is bounded by $Δ$.   In time $ O(m⋅Δ)^m ⋅ \|b\|_1 ⋅ n $
   one can compute  a solution of
  \begin{displaymath}
    Ax = b, \, x ∈ ℤ^n_{≥0}
  \end{displaymath}
  or assert that such a solution does not exist. 
\end{theorem}

We next describe how to tackle the optimization problem~\eqref{eq:2}. 
We introduce weights on the arcs of the digraph $D = (\mathscr{S},A)$. The weight of the arc $xy$ is  $c_i$ if $y-x$ is the $i$-th column of $A$. Down below, we will argue that the longest path in the thereby weighted digraph from $0$ to $b$ corresponds to an optimal solution of \eqref{eq:2}. The longest path problem in
 $D$ can be solved in time $O(|\mathscr{S}|⋅|A|)$ with the Bellman-Ford algorithm~\cite{AHU74}. Since 
 \begin{displaymath}
   |A| ≤ |\mathscr{S}| ⋅ n  
 \end{displaymath}
our discussion below  implies that the integer program~\eqref{eq:2} can be solved in time $O(n ⋅ |\mathscr{S}|^2)$   provided that there do not exist positive cycles reachable from $0$. The next lemma clarifies that such a positive cycle exists if and only if the feasible  integer program~\eqref{eq:2} is unbounded.

\begin{lemma}
  \label{lem:2}
  Suppose that~\eqref{eq:2} is feasible. The integer program~\eqref{eq:2} is unbounded if and only if $D$ contains a cycle of strictly positive length that is reachable from $0$. 
\end{lemma}

\begin{proof}
  It follows from the theory of integer linear programming~\cite{Schrijver86} that~\eqref{eq:2} is unbounded if and only if there exists an  integer solution  of $Ax = 0, \, x≥0,\, c^Tx >0$. Let $r^* ∈ ℤ_{≥0}^n$ be such a solution. Using the Steinitz-type-lemma in the spirit of the rearrangement~\eqref{eq:3} but with $b=0$, $r^*$ corresponds to a (not necessarily simple) cycle in $D$  of positive length starting at $0$. This proves the lemma. 
\end{proof}

\begin{remark}
\label{rem:1}
  The reader might have noticed that $D$
  contains a positive simple cycle that is reachable from $0$
  if and only if there exists a positive simple cycle in $D$
  containing $0$.
  The two cycles however might not be a translation of each other.
\end{remark}

The algorithm to solve~\eqref{eq:2} is now as follows. We first check integer feasibility of~\eqref{eq:2}. Then we run a single-source longest path algorithm from $0$ to the other nodes of $D$, in particular to $b$. If the algorithm detects a cycle of positive weight, we assert that \eqref{eq:2} is unbounded. Otherwise, the longest path form $0$ to $b$  corresponds to an optimal solution of~\eqref{eq:2}. We therefore have proved the following theorem.

\begin{theorem}
  \label{thr:5}
  The integer program~\eqref{eq:2} can be solved in time
  \begin{displaymath}
      n ⋅ O(m⋅Δ)^{2⋅m} ⋅ \|b\|_1^2 %= (m ⋅ Δ)^{O(m)} ⋅ \|b\|_∞^2
  \end{displaymath}
  where $Δ$
  is an upper bound on all absolute values of entries in $A$.
\end{theorem}

\begin{remark}
  \label{rem:2}
  \begin{enumerate}[a)]
  \item For an integer program in standard form, without upper bounds on the variables, we can assume that $A$ does not have repeated columns. Hence $n = O(Δ^m)$. The running time bound in Theorem~\ref{thr:5} is $(m ⋅ Δ)^{O(m)} ⋅ \|b\|_∞^2$. 
  \item The longest path problem runs in linear time if the digraph $D$ does not have any cycles at all. This is for example the case when $A$ has only non-negative entries. In this case one has a running time of  $O(m⋅Δ)^m ⋅ \|b\|_1 ⋅ n$ for the integer program~\eqref{eq:2}. 
A well known example of such an integer program is the \emph{configuration IP} for scheduling, see, e.g. \cite{Hochbaum:1987:UDA:7531.7535,jansen2010eptas,jansen_et_al:LIPIcs:2016:6212}.
% For $m=1$, this matches the running time of dynamic programming for the \emph{unbounded knapsack problem}, see, e.g.~\cite{martello1990knapsack} c.f. Section~\ref{sec:bound-knaps-probl}.  
\end{enumerate}
\end{remark}

\begin{remark}
  \label{rem:4}
  For the case in which $Δ$
  is an upper bound on the absolute values of the entries of both $A$
  and $b$ the set $\mathscr{S}$ contains at most $(4\,m ⋅ Δ+1)^m$ elements and the integer program \eqref{eq:2} can be solved in time $n ⋅ O(m ⋅ Δ)^{2\,m}$ and in time  $n ⋅ O(m ⋅ Δ)^{m}$ if the digraph does not have any cycles. 
\end{remark}

% \subsection{Nonnegative weights} 
% \label{sec:unbo-knaps-probl}
  
% \todo{I am not sure that we should say something on this now}

% The unbounded knapsack problem is an integer program of the form~\eqref{eq:2}, where $m = 1$ and each entry of $A,b$ and $c$ is non-negative.  If $Δ$ is an upper bound on the entries of $A$ and $b$ then dynamic programming yields a running time of $O(n ⋅ Δ)$, see~\cite{martello1990knapsack,kellerer2004knapsack}. If we plugin our bound from Theorem~\ref{thr:5} with $m=1$  we obtain a running time of $O(n ⋅ Δ^2)$. Yet, the shortest path computation that we described above can be sped up since the arc-weights are non-negative.  First of all, we can remove any variable 

% I our setting we recall that $|V| = O((m ⋅Δ)^m)$ and $|A| = O(n ⋅ |V|)$.  If 
% $  \log \log |V| ≥ n $ then $|V| ≥ 2^{2^n}$ which would imply that the number of nodes alone is an upper bound on the $n^{O(n)}$ running time of general integer programming with $n$ variables~\cite{Kannan87}. Therefore we can note the following theorem. 

% \begin{theorem}
%   \label{thr:4}
%   If all weights $c_i$
%   in the integer program~\eqref{eq:2} are positive, then an optimal
%   solution can be found in time $O(n ⋅ Δ^{m})$
%   where $Δ$
%   is an upper bound on all absolute values of entries in $A$
%   and $b$
%   and both $A$
%   and $b$ are integer and $m$ is the number of rows of $A$.
% \end{theorem}

\section{Proximity in the $\mathbf{\ell_1}$-norm} 
\label{sec:proximity:-l_1-norm}

In this section, we  provide the results \ref{item:5}) and 
\ref{item:7}). 
From now on we consider integer programs in standard form with upper bounds on the variables, where the absolute values of $A$ only need to be bounded by some integer $Δ$.  In other words, we consider a problem of the form 
\begin{equation}
\label{eq:6}
  \max \{ c^T x : Ax = b, \, 0 ≤x ≤u , \, x ∈ ℤ^n \}  
\end{equation} 
where $A ∈ ℤ^{m×n}$, $b ∈ℤ^m$ and $c ∈ ℤ^n$ and $u ∈ℕ^n$ such that $|a_{ij}|≤Δ$ for each $i,j$. 
We are interested in the distance between an optimal vertex of the LP-relaxation of~\eqref{eq:6} and a closest integer optimum $z^*$ in the $\ell_1$-norm. 

\paragraph{A previous bound} that has been useful in many algorithmic applications, see for example~\cite{sanders2009online} was shown by Cook et al.~\cite{CookGerardsSchrijverTardos86}. In its full generality, it is concerned with  the distance in the $\ell_∞$  norm in the setting of an integer program in inequality form 
\begin{equation}
  \label{eq:4}
  \max \{ c^Tx : Ax ≤b, \, x ∈ℤ^n\}. 
\end{equation}
We suppose that $A$ and $b$ are integral and that \eqref{eq:4} is feasible and bounded. 
 Cook et al.~\cite{CookGerardsSchrijverTardos86} show that for any optimal solution $x^*$ of the linear programming relaxation there exists an optimal solution $z^*$ of the integer program with 
 \begin{equation}
   \label{eq:5}
   \| x^* - z^* \|_∞ ≤ n ⋅ δ,
 \end{equation}
where $\delta $ is the largest absolute value of the determinant of any square submatrix of $A$. By the Hadamard bound, see, e.g.\cite{Schrijver86}, $δ$ is bounded by $n^{n/2} ⋅Δ^n$, where $Δ$ is, as before, an upper bound on the absolute values of the entries of $A$. 

Applied to an integer program in standard form~\eqref{eq:2} this
result implies that, for a given optimal linear solution $x^*$
there exists an integer optimal solution $z^*$
such that $\|z^*-x^*\|_1 \leq n^2 δ$.
Since the Hadamard bound implies $δ ≤m^{m/2}Δ^m$
\begin{equation}
  \label{eq:15}
  \|z^*-x^*\|_1 ≤ n^2 ⋅ m^{m/2}Δ^m. 
\end{equation}

\paragraph{Using the Steinitz lemma,}  we   show next that
\begin{displaymath}
  \|z^*-x^*\|_1 ≤ m ⋅ (2⋅m⋅Δ+1)^m.  
\end{displaymath}
We will see in a later section how this leads to  algorithms for integer programs in standard form with upper bounds on the variables. 
In the following, let $x^*$ and $z^*$ be  optimal solutions of the linear programming relaxation of \eqref{eq:6} and of the integer program \eqref{eq:6}  respectively. A vector $y ∈ℤ^n$ is called a \emph{cycle} of $(z^* - x^* )$ if $A \,y  = 0$ and 
\begin{equation}
  \label{eq:7}  
   |y_i| ≤ |(z^* - x^*)_i| \,\text{ and }\, y_i ⋅(z^* - x^*)_i ≥0\, \text{ for each }  i. 
 \end{equation}
 
\begin{lemma}
  \label{lem:3}
  Let $y$ be a cycle of $(z^* - x^* )$, then the following assertions hold. 
  \begin{enumerate}[i)]
  \item $z^* - y$ is a feasible integer solution of \eqref{eq:6}. \label{item:1}
  \item  $x^* +  y$ is a feasible solution of the linear programming relaxation of \eqref{eq:6}. \label{item:2}
  \item One has $c^T y ≤0$. \label{item:3}
  \end{enumerate}
\end{lemma}

\begin{proof}
  We show \ref{item:1}) and \ref{item:2}). Since $A\,y=0$ and $y$ is integral, 
  we only need to verify that the bounds on the variables
  are satisfied.

  If $(z^* - x^*)_i <0$, then $y_i ≤0$ and since $z^*$ and $x^*$ are feasible, one has 
  \begin{displaymath}
    0 ≤ z_i^* - y_i  ≤  z_i^* - (z^* - x^*)_i 
                 =  x^*_i 
                 ≤  u_i
  \end{displaymath}
              and 
              \begin{displaymath}               
u_i ≥    x_i^* + y_i  ≥   x_i^* + (z^* - x^*)_i 
                 =  z^*_i 
                 ≥  0.
              \end{displaymath}
              If $(z^* - x^*)_i > 0$   is analogous.         
%  and we only need to verify that the upper  bound $ z^*_i - y_i≤u_i$ is not violated. But one has $y_i ≥ (z^* - x^*)_i$ which is equivalent to $ x^*_i ≥ (z^* - y)_i$. Since $x^* \leq u$ one has $ z^*_i - y_i≤u_i$. 
%The case where   $(z^* - x^*)_i ≥0$ follows by a similar line of argument. % then $y_i ≥0$ and we only need to verify that the lower  bound $ 0 ≤ z^*_i - y_i$ is not violated. But in this case we have $y_i ≤ (z^* - x^*)_i$ which is equivalent to  $x^*_i ≤ z^*_i - y_i$ and this implies $ 0 ≤ z^*_i - y_i$ since $x^* ≥0$. 

To see~\ref{item:2}) note that $y_i>0$ implies that $z^*_i > x^*_i$ and thus $x^*_i$ is not at the upper bound $u_i$. If $y_i<0$ then $z^*_i < x^*_i$ which means that the lower bound $0≤x_i$ is not tight at $x^*$. Therefore, there exists an $ε>0$ such that $x^* +ε y$ is a feasible solution of the linear program. 

The assertion~\ref{item:3}) follows from the optimality of $x^*$ and~\ref{item:2}). 
\end{proof}

\begin{lemma}
  \label{thr:4}
  Let $x^*$
  be an optimal solution of the linear programming relaxation
  of~\eqref{eq:6} and let $z^*$
  be an optimal integer solution of \eqref{eq:6} such that
  $\|z^*-x^*\|_1$ is minimal. There does not exist a cycle of   $z^*-x^*$. 
\end{lemma}

\begin{proof}
  Suppose that $y$ is a cycle of  $z^*-x^*$. By \ref{item:1}) and \ref{item:3}) of Lemma~\ref{lem:3},  $z^* - y$ is also an optimal solution of the integer program~\eqref{eq:6}.  But $\|z^* - y - x^* \|_1 < \|z^* -  x^* \|_1$ contradicting the minimality of $\|z^* -  x^* \|_1$. 
\end{proof}

We are now ready to apply the Steinitz-type lemma to derive a new bound on the  $\ell_1$-distance  between $x^*$ and $z^*$. 

\begin{theorem}
	\label{thr:2}
        Let $x^*$ be an optimal vertex solution of the linear programming relaxation of~\eqref{eq:6}. 
	There exists  an optimal solution $z^*$ of the integer program~\eqref{eq:6} such that 
        \begin{displaymath}
          	\|z^*-x^*\|_1 \leq m⋅ (2\, m⋅ Δ+1)^m .
        \end{displaymath}
        Here, $Δ$ is an upper bound on the absolute values of the entries in $A$.         
\end{theorem}

\begin{proof}
  Let $z^*$ be an optimal integer solution such that $\|z^* - x^*\|_1$ is minimal. In the following we use the notation $⌊x^*⌉$ for the vector that one obtains from  $x^*$ by rounding each component towards the corresponding component of $z^*$. More precisely, the $i$-th component of $⌊x^*⌉$ is set to 
  \begin{displaymath}
⌊x^*⌉_i =     \begin{cases}
      ⌈x^*⌉_i & \text{ if } z^*_i > x^*_i \text{ and } \\
      ⌊x^*⌋_i & \text{ if } z^*_i ≤ x^*_i 
    \end{cases}    
  \end{displaymath}
and we denote the rest by $\{x^*\} = x^* - ⌊x^*⌉$.  Clearly, one has 
  \begin{equation}
    \label{eq:9}
    A (z^* - ⌊x^*⌉) - A\{x^*\} = 0. 
  \end{equation}
  We are now again in the setting of the Steinitz-lemma where we have a sequence of vectors 
  \begin{equation}
    \label{eq:11}
    v_1,\dots,v_t,- A\{x^*\}
  \end{equation}
that sum up to zero. More precisely this sequence is constructed as follows. Start with the empty sequence. For each column index $i$ append $|(z^* - ⌊x^*⌉)_i|$ copies of  $\sign((z^* - ⌊x^*⌉)_i) ⋅ a_i$ to the list, where $a_i$ is the $i$-th column of $A$. Finally append $-A\{x^*\}$ to the list. 
Since $x^*$ has at most $m$ positive entries, we conclude that $\|-A\{x^*\}\|_∞ ≤ Δ ⋅m$ and that there are integer vectors $w_1,\dots,w_m$ of $\ell_∞$-norm at most $Δ$ with 
\begin{displaymath}
  -A\{x^*\} = w_1+ \cdots +w_m. 
\end{displaymath}
This means that the sequence of vectors~\eqref{eq:11} can be expanded to a sequence
\begin{equation}
  \label{eq:13}
  v_1,\dots,v_t,w_1,\dots,w_m
\end{equation}
where each vector is at most of $\ell_∞$-norm $Δ$ and that sum up to the zero vector. Observe that $t = \|z^* - ⌊x^*⌉\|_1$ and that $t+m ≥\|z^* - x^*\|_1$. 
The Steinitz Lemma implies that the sequence~\eqref{eq:13} can be re-arranged in such a way 
\begin{equation}
  \label{eq:16}
  u_1,\dots,u_{t+m}
\end{equation}
that for each $1 ≤ k ≤ t+m$ the partial sum $p_k = ∑_{i=1}^k u_i$ satisfies 
\begin{equation}
  \label{eq:18}
\|p_k\|_∞  ≤ m  Δ. 
\end{equation}
We will now argue that there cannot be indices  $1≤ k_1 < \cdots < k_{m+1} ≤ t+m$ with 
\begin{equation}
  \label{eq:19}
  p_{k_1} = \cdots = p_{k_{m+1}}, 
\end{equation}
which implies that $t+m$ is bounded by $m$ times the number of integer points of norm at most $m  ⋅ Δ$ and therefore 
\begin{displaymath}
  \|z^* - x^*\|_1 ≤ m ⋅ (2 ⋅ m⋅Δ +1)^m. 
\end{displaymath}
Assume to this end that there exist $m+1$ indices $1≤ k_1 < \cdots < k_{m+1} ≤ t+m$ satisfying~\eqref{eq:19}.
This yields a partition of the sequence into $m+1$ nonempty pieces that sum up to zero, namely:
\begin{displaymath}
  u_1,\dots,u_{k_1}, u_{k_{m+1}+1},\dots,u_{t+m} 
  \end{displaymath}
  and
  \begin{displaymath}
    u_{k_j+1},\dots,u_{k_{j+1}}, \,  j=1,\dots,m. 
  \end{displaymath}
One of these subsequences does not contain an element from $\{w_1,\dots,w_m\}$, and hence 
 are columns of $A$ or negatives thereof.  This corresponds to a cycle $y$ of $z^* - x^*$ which, by the minimality of $\|z^* - x^*\|_1$ and  Lemma~\ref{thr:4} is impossible. 

% If such an index $k_i$ does not exist, then all the vectors $w_1,\dots,w_m$ appear in the sequence 
% \begin{displaymath}
%   u_{k_{1}+1},\dots,u_{k_{m+1}} .
% \end{displaymath}
% One of these subsequences does not contain an element in $\{w_1,\dots,w_m\}$ and this corresponds to a cycle $y$ of $z^* - x^*$ which, by the minimality of $\|z^* - x^*\|_1$ and  Lemma~\ref{thr:4} is impossible. 
%This corresponds to an integer vector $y ∈ ℤ^n$ such that  $A (y - \{x\})=0$, $|(y - \{x\})_i | ≤ |(z^* -x^*)_i|$    and  
%$  \sign(y_i - \{x\}_i) = \sign({z^*}_i -{x^*}_i) \text{ for each } i$. 
%This implies that $z^* - x^* - (y-\{x\})$ is a cycle of $z^* - x^*$ which is again impossible. 
% \begin{displaymath}
%   \left \| ∑_{i=1}^k v_{π(i)} \right \|_∞ ≤m (m⋅ Δ-1)
% \end{displaymath}
% holds for each $k=1,\dots,t+1$. By removing  the vector $v_{t+1}$ we conclude that there exits a permutation $π$ of $1,\dots,n$ such that 
% \begin{equation}
%   \label{eq:12}
%     \left \| ∑_{i=1}^k v_{π(i)} \right \|_∞ ≤m^2⋅ Δ + m⋅Δ
% \end{equation}
\end{proof}

\subsection{Integrality gaps of integer programs} 
\label{sec:integr-gaps-integ}

Our bound of Theorem~\ref{thr:2} directly leads to a bound on the \emph{(absolute) integrality gap} of integer programs.  This gap is  $c^T(x^*-z^*)$ and 
 can, via Theorem~\ref{thr:2}, be bounded by 
 \begin{equation}
   \label{eq:17}
   c^T(x^*-z^*) \leq \|c\|_\infty \|z^*-x^*\|_1 \leq \|c\|_\infty m ⋅ (2 ⋅ m ⋅Δ+1)^m  .
 \end{equation}
An integer program~\eqref{eq:2} is called an \emph{unbounded knapsack problem} if $m$ =1. In this case, 
Aliev et al.~\cite{aliev2016integrality}  show that 
one has 
\begin{equation}
  \label{eq:28}
  c^T(x^*-z^*) \leq 2⋅\|c\|_∞ ⋅Δ
\end{equation}
which is asymptotically our bound for $m=1$. They derived their bound using methods from the geometry of numbers.  A careful analysis of our proof in the case $m=1$ also yields the  bound~\eqref{eq:28} exactly. More precisely, this follows since we can choose $u_1 = w_1$  in the Steinitz sequence~(\ref{eq:16}). This is special about the one-dimensional case, any vector can be chosen as the first element. Clearly, $w_1$ cannot be re-visited as a partial sum. This implies $\|z^* -x^*\|_1 ≤(2 Δ +1) -1 = 2Δ$. 

% \begin{remark}
% In the case $m=1$  our bound~\eqref{eq:17}  is best possible as the following example shows. 
% For a  natural number $a$ and a rational value $\epsilon \in (0,1)$, consider the optimization problem 
% $$\max\{(a+1 + \epsilon )x + ay \text{ s.t. } (a+1)x + ay = a^2 = (a-1)(a+1) +1,\; x,y \in \Z_+\}.$$
% The unique optimal solution, $(x^*,y^*)$  of the linear relaxation is equal to $(a-1 + \frac{1}{a+1},0)$. 
% The unique integer optimal solution is equal to $(0,a)$. The value $\|(x^*,y^*) - (0,a)\|_1 \geq 2a-1 = 2(a+1)-3.$	
% \end{remark}

\section{Algorithmic implications}
\label{sec:integ-progr-with}

We now devote our attention to dynamic programming algorithms for integer programs in standard form with upper bounds on the variables and where $|a_{ij} |≤ Δ$ for each $i,j$. This setting has received  considerable attention in the approximation algorithm community, especially for scheduling problems and the respective \emph{configuration LPs}, see for example \cite{sanders2009online,jansen2010eptas,jansen_et_al:LIPIcs:2016:6212}. % The classical  approximation scheme for makespan scheduling relies on the solution of an integer program of the form~\eqref{eq:2} where $m$ and $Δ$ are bounded by $1/ε$. 

Our proximity result can now be used in a dynamic programming approach to solve an integer program in standard form with upper bounds on the variables~\eqref{eq:6}. We first compute an optimal basic solution $x^*$ of the LP-relaxation of~\eqref{eq:6}. 
In the following we denote our bound on  $\|z^* - x^*\|_1$ by $L'_1 = m ⋅ (2⋅ m⋅Δ+1)^m$. 
Theorem~\ref{thr:2} reveals that there exists an optimal integer solution $z^*$ with 
\begin{displaymath}
  \|z^* -  ⌊x^*⌋ \|_1 ≤ \|z^* -  x^* \|_1 + \|x^* -  ⌊x^*⌋ \|_1 ≤  L'_1 +m =: L_1.
\end{displaymath}
After the variable transformation $y = z -  ⌊x^*⌋$ one has  to solve an integer program of the form
\begin{equation}
  \label{eq:20}
  \begin{array}{lll}
    \max c^T y & \text{ s.t. } \vspace{.1cm} \\
&    A \, y  =  A ⋅ \{x^*\} \\
 &   -l^* ≤ y ≤ u^*\\
  &  \|y\|_1 ≤ L_1  \\
  &  y ∈ ℤ^n
  \end{array}
\end{equation}
where $l^* = \min\{ L_1, ⌊x^*⌋\}$ and $u^* = \min\{ L_1 , u - ⌊x^*⌋\}$. Notice that $\|l^*\|_∞ ≤ L_1$ and $\|u^*\|_∞ ≤ L_1$. The potential of the new proximity bound lies in the constraint on the $\ell_1$-norm. 
For   $y ∈ℤ^n$ that satisfies  $\|y\|_1 ≤ L_1$ one has for each $1 ≤ k ≤ n$ 
\begin{equation}
  \label{eq:22}
  \|∑_{i=1}^k y_i ⋅ a_i  \|_∞ ≤ Δ ⋅  L_1.  
\end{equation}
Let $U ⊆ ℤ^m$ be the set of integer vectors of infinity norm at most $Δ ⋅  L_1 $. The cardinality of $U$ is equal to  
\begin{equation}
  \label{eq:23}
  |U | = (2⋅ Δ \,L_1+1)^m = O(m ⋅ Δ)^{m⋅(m+1)}. 
\end{equation}
To find the optimal $y^*$ we build the following acyclic directed graph, see Figure~\ref{fig:2}. The nodes of the graph consist of  a starting node $s=0$  and a target node $t = A⋅\{x^*\}$. Furthermore, we have $n-1$ copies of the set  $U$ that we denote by $U_1,...,U_{n-1}$. The arcs are as follows. 

There is an arc from $s$ to a node $v ∈ U_1$ if there exists an integer $y_1$ such that 
\begin{displaymath}
   v = y_1⋅ a_1 \text{ and }  -l^*_1 ≤y_1≤u^*_1
\end{displaymath}
holds. Again, $a_1$ denotes the first column of $A$. The weight of the arc is  $c_1⋅y_1$.  
There is an arc from a  node $u ∈ U_{i-1}$ to a node  $v ∈ U_{i}$ if there exists an integer $y_i$ such that 
\begin{displaymath}
  v-u = y_i ⋅a_i \text{ and } -l^*_i ≤y_i≤u^*_i
\end{displaymath}
holds. The weight of this arc is $c_i ⋅ y_i$. 
Finally, there is an arc from $u ∈ U_{n-1}$ to $t$ of weight $y_n ⋅ c_n$ if 
\begin{displaymath}
  A\{x^*\} - u = y_n ⋅a_n \text{ and } -l^*_n ≤y_n≤u^*_n
\end{displaymath}
holds for some integer $y_n$. 
Clearly, a longest path in this graph corresponds to an optimal solution $y^*$ of the integer program~\eqref{eq:6}. 
The out-degree of each node is bounded by $u_i^* + l_i^* ≤ 2 ⋅ L_1 +1$. Therefore,  
the number of arcs is bounded by 
\begin{equation}
  \label{eq:24}
  n ⋅ |U| ⋅ (2 ⋅ L_1 +1 ) = n ⋅ O(m)^{(m+1)^2} ⋅ O(Δ)^{m ⋅ (m+2)} 
\end{equation}
which would lead to a corresponding running time of $n ⋅ O(m)^{(m+1)^2} ⋅ O(Δ)^{m ⋅ (m+2)}$ since longest path in an acyclic digraph can be computed in linear time in the number of nodes and arcs. 

However, a standard technique can be applied to significantly decrease
the number of arcs. This idea is based on the binary representation of
an integer and is as follows. Imagine that, for each interval $[-L,U]$ with $L,U ∈ ℕ$,  
there exist a number $k = O(\log^2 (U+L))$ and integers
\begin{displaymath}
  s_1,\dots,s_k 
\end{displaymath} such that
\begin{enumerate}[i)]
\item For each $z ∈ [-L,U]$ there exist $y_1,\dots,y_k ∈ \{0,1\}$ such that \label{item:8}
  \begin{displaymath}
    ∑_{j=1}^k y_j ⋅s_j = z. 
  \end{displaymath} 
\item  For each choice of $y_1,\dots,y_k ∈ \{0,1\}$ one has  \label{item:9}
  \begin{displaymath}
    ∑_{j=1}^k y_j ⋅s_j ∈ [-L,U].
  \end{displaymath} 
\end{enumerate}
We can then replace the part of the digraph connecting  $U_{i-1}$ and $U_i$ with $O(\log^2 (l_i^* + u_i^*))$  copies of $U$. 
Each copy is associated to a binary variable $y^{i}_j$ and an integer $s^{i}_j$ corresponding to the construction for the interval $[-l_i^*,u_i^*]$.   We order them arbitrarily and have an arc from a node $u$ from one copy of $U$ to the node $v$ of its successor of weight zero, if $u=v$ and of weight 
$c_i ⋅ s^{i}_j$ if the successor copy is associated to the variable $y^{i}_j$ and $v = u + a_i ⋅ s^{i}_j$. 
In this way, the out-degree of each node is at most two and the total number of nodes and arcs is 
\begin{displaymath}
  n ⋅ O(\log^2 L_1) ⋅ |U| = n\,m \,   O( \log m⋅Δ  ⋅ (m \, Δ)^{m ⋅(m+1)}), 
\end{displaymath}
 where  we assume $Δ≥2$. 
\begin{figure*} 
  \begin{center}
    \includegraphics[height=1.8cm]{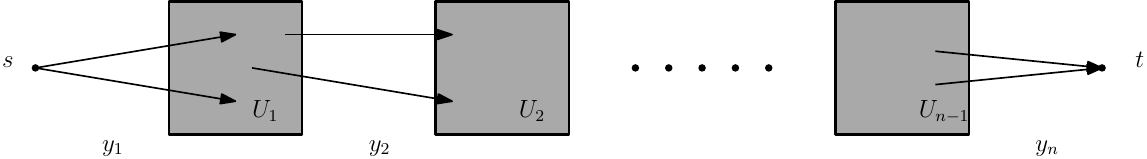}
  \end{center}
  \caption{An illustration of the directed acyclic graph to solve the integer program~\eqref{eq:6}.}
   \label{fig:2}
  \end{figure*}
We therefore have the following result. 
\begin{theorem}
  \label{thr:6}
  An integer program of the form~\eqref{eq:6} can be solved in time
  \begin{displaymath}
    n ⋅  O(m)^{(m+1)^2} ⋅ O(Δ)^{m⋅(m+1)} ⋅ \log^2 (m⋅Δ)
  \end{displaymath}
  if each component of $A$ is bounded by $Δ$ in absolute value. 
\end{theorem}

Let us briefly comment on how to find  these integers $s_1,\dots,s_k$  satisfying~(\ref{item:8}) and  (\ref{item:9}).  It is enough to show how to find them for an interval of the form $[0,U]$, since 
\begin{displaymath}
  [-L,U] = \{-x +  y : x ∈ [0,L], \, y ∈ [0,U]\}. 
\end{displaymath}
Thus, let $[0,U]$ be a given interval. 
 If $U = 2^k-1$ one lets $s_j = 2^{j-1}$ for $j=0,\dots,k-1$. If $U$ is not of this form, then let $p$ be the largest power of two less than or equal to $U$. For the interval $[0,p-1]$ we use the construction from above. Now we are left with representing the interval $[0,U-p+1]$ and concatenate the sequence of integers $s_j$ from both constructions. Since the interval $[0,U-p+1]$ is half as long as $[0,U]$, an inductive argument applies and the conditions ~(\ref{item:8}) and  (\ref{item:9}) are satisfied.

\subsection{Faster algorithms for  integer knapsack} 
\label{sec:bound-knaps-probl}

The \emph{ bounded  knapsack problem} is of the following kind  
\begin{equation}
  \label{eq:14}
  \max\{c^Tx : a^Tx = β, \,  0≤x≤ u, \, x ∈ ℤ^n\} 
\end{equation}
where $c,a,u ∈ ℤ_{>0}^n$ and $β ∈ ℤ_{>0}$.  If the upper bound is $u=β ⋅ \mathbf{1}$, then the knapsack problem is called  \emph{unbounded}. We let $Δ_a$ be an upper bound on the entries of $a$. 

Tamir~\cite{tamir2009new} has shown that the unbounded and bounded knapsack problem can be solved in time $O(n^2 Δ_a^2)$ and in time $O(n^3Δ^2_a)$ respectively. These running times were obtained by applying the proximity result of Cook et al.~\cite{CookGerardsSchrijverTardos86}. We now 
use our proximity bound to save a factor of $n$ in each case. 

\subsubsection*{Unbounded knapsack}
We begin with the unbounded knapsack problem. An optimal fractional vertex  $x^*$ has only one positive entry, $x^*_1$ lets say and by Theorem~\ref{thr:2} there exists an optimal integer solution $z^*$ with $\|z^* - x^*\|_1 ≤ 2⋅ Δ_a +1$. 
We can assume that $x^*_1≥  2⋅ Δ_a +1$ since otherwise $β = O(Δ_a^2)$ and an $O(n ⋅Δ_a^2)$ algorithm is obvious, see  Remark~\ref{rem:4}. If $y^*$ is an optimal solution of 
\begin{equation}
  \label{eq:21}
  \max\{ c^Ty : a^Ty = (2⋅ Δ_a+1 + \{x_1^*\}) a_1, \, y≥0,\, y ∈ℤ^n\},
\end{equation}
then $(y^*_1+⌊x^*_1⌋-(2⋅ Δ_a+1),y^*_2,\dots,y^*_n)$ is an optimal solution of the unbounded knapsack problem. Since all entries of $a$ and $(2⋅ Δ_a+1) a_1$ are positive and bounded by $O(Δ_a^2)$ one can solve the knapsack problem~\eqref{eq:21} in time $O(n ⋅Δ_a^2)$, see again Remark~\ref{rem:4} and notice that the digraph has no cycles as all integers are positive. Consequently we have the following theorem. 
\begin{theorem}
  \label{thr:7}
  An unbounded  knapsack problem~\eqref{eq:14} 
can be solved in time $O(n⋅Δ_a^2)$. 
\end{theorem}
% The higher running time of Tamir's algorithm stems from the extra factor of $n$ in the proximity bound  of $\ell_∞$ of  Cook et al.~\cite{CookGerardsSchrijverTardos86}. 

\subsubsection*{Bounded knapsack}

Setting  $m=1$ in Theorem~\ref{thr:6} we obtain a running time of 
\begin{displaymath}
  O(n ⋅ (\log Δ)^2 ⋅ Δ^2).  
\end{displaymath}
which is already an improvement over the running time of Tamir's algorithm if $\log Δ ≤n$. A running time of $O(n^2 ⋅Δ^2)$ can be obtained as follows. Again, we solve the linear programming relaxation of~\eqref{eq:14} and obtain an optimal vertex solution $x^*$. Following the notation from Section~\ref{sec:integ-progr-with} we now have to solve an integer program of the form
\begin{equation}
  \label{eq:26}
  \max\{ c^T x : a^Tx = β', \, -l^* ≤ x ≤ u^*, \, x ∈ ℤ^n\},
\end{equation}
where $β'$ is an integer with $0 ≤β'≤Δ_a$ and $\|l^*\|_∞, \|u^*\|_∞ ≤ 2⋅Δ_a+1$. This is equivalent to the bounded knapsack problem 
\begin{equation}
  \label{eq:27}
  \max\{ c^T x : a^Tx = β'+ ∑_i a_i⋅l_i^*, \, 0 ≤ y ≤ l_i^*+u^*, \, x ∈ ℤ^n\}.
\end{equation}
The new right-hand-side of this problem is $O(n ⋅Δ_a^2)$.
Pferschy~\cite{pferschy1999dynamic} has shown that a bounded knapsack problem in $n$ variables and right-hand-side $γ$ can be solved in time $O( n ⋅ γ)$. Thus the 
bounded knapsack problem can  be solved in time $O(n^2 ⋅ Δ_a^2)$.

\subsection*{Acknowledgments}

We would like to thank Janos Pach for pointing us to the papers~\cite{sevast1997steinitz,sevastianov1978approximate}. We also thank Daniel Dadush for his very helpful comments, in particular for his hint on a variation of the Steinitz theorem that lead to stronger running-time bounds.
We are grateful to the anonymous referees for several suggestions that led to an improved version of this manuscript. 
We thank  Martin Skutella for hosting us  at TU Berlin and EPFL for hosting the second author. The first author acknowledges support from the Swiss National Science Foundation (SNSF) within the project \emph{Lattice Algorithms and Integer Programming (Nr. 185030)}. The second author acknowledges support from the \emph{Alexander von Humboldt Foundation}.

\bibliographystyle{abbrv}
% \bibliography{/Users/eisenbrand/Dropbox/TeX/Bib/books,/Users/eisenbrand/Dropbox/TeX/Bib/mybib,/Users/eisenbrand/Dropbox/TeX/Bib/papers,/Users/eisenbrand/Dropbox/TeX/Bib/my_publications}

\end{document}